\documentclass[journal, 12]{IEEEtran}
\usepackage[T1]{fontenc}
\usepackage[latin9]{inputenc}
\usepackage{amsmath}
\usepackage{graphicx}
\usepackage{amssymb}
\usepackage{esint}
\makeatletter
\usepackage{ntheorem}
\usepackage{etex}
\usepackage{amsfonts}
\usepackage{multirow}
\usepackage{pictex}
\usepackage{mathtools}
\usepackage{relsize}

\usepackage{cite}
\newtheorem{lemma}{Lemma}
\newtheorem{theorem}{Theorem}
\newtheorem*{unnumbered}{Rate-distortion Theorem:}
\newtheorem{corollary}{Corollary}
\makeatother

\begin{document}
\title{On Constrained  Randomized Quantization}
\author{Emrah~Akyol,~Student Member,~IEEE, and Kenneth~Rose,~Fellow,~IEEE
\thanks{Authors are with the Department of Electrical and Computer Engineering,
University of California, Santa Barbara, CA, 93106 USA, e-mail: \{eakyol,rose\}
@ece.ucsb.edu.%
}
\thanks{This work is supported by the NSF under the grants CCF-0728986, CCF-1016861 and CCF 1118075. The material in this paper was presented in part at the
IEEE Data Compression Conferences, March 2009 and April 2012. %
} }

\markboth{submitted to IEEE Transactions  on Signal Processing}%
{{Akyol, Rose}: On Optimal Randomized Quantization}

\maketitle
\begin{abstract}
Randomized (dithered) quantization is a method capable of achieving
white reconstruction error independent of the source. Dithered quantizers
have traditionally been considered within their natural setting of
uniform quantization. In this paper we extend conventional dithered
quantization to nonuniform quantization, via a subterfage: dithering
is performed in the companded domain. Closed form necessary conditions
for optimality of the compressor and expander mappings are derived
for both fixed and variable rate randomized quantization. Numerically,
mappings are optimized by iteratively imposing these necessary conditions.
The framework is extended to include an explicit constraint that deterministic
or randomized quantizers yield  reconstruction error that is uncorrelated
with the source. Surprising theoretical results show direct and simple
connection between the optimal constrained quantizers and their unconstrained
counterparts. Numerical results for the Gaussian source provide strong
evidence that the proposed constrained randomized quantizer outperforms
the conventional dithered quantizer, as well as the constrained deterministic
quantizer. Moreover, the proposed constrained quantizer renders the reconstruction error nearly white.  In the second part of the paper, we investigate whether
uncorrelated reconstruction error requires random coding to achieve
asymptotic optimality. We show that for a Gaussian source, the optimal
vector quantizer of asymptotically high dimension whose quantization
error is uncorrelated with the source, is indeed random. Thus, random
encoding in this setting of rate-distortion theory, is not merely
a tool to characterize performance bounds, but a required property
of quantizers that approach such bounds. 
\end{abstract}
\begin{keywords} Source coding, dithered quantization, subtractive dithering, compander, quantizer design, analog mappings.
\end{keywords}

\section{Introduction}

\label{sec:intro}

Dithered quantization is a randomized quantization method introduced in \cite{roberts1962picture}. A central motivation for dithered quantization is its ability to yield
quantization error that is independent of the source, which
can be achieved if certain conditions, determined by Schuchman, are
met \cite{Schuchman}. Traditionally, dithered quantization has been
studied in the framework where the quantizer is uniform (with step
size $\Delta$) and the dither signal is uniformly distributed over
$(-\frac{\Delta}{2},\frac{\Delta}{2})$, matched to the quantizer
interval as shown in Figure 1. A uniformly distributed dither
signal is added before quantization and the same dither signal is
subtracted from the quantized value at the decoder side. Note that
only subtractive dithering is considered in this paper. In the variable
rate case, the quantized values are entropy coded, conditioned on
the dither signal. Randomized (dithered) quantizers have been studied
in the past due to important properties that differentiate them from
deterministic quantizers, and were employed to characterize rate-distortion
bounds for universal compression \cite{ziv_quantization,graydq}.
Zamir and Feder provide extensive studies
of the properties of dithered quantizers \cite{zamiruqr,zamirirp}.

\begin{figure*}
\includegraphics[scale=0.30]{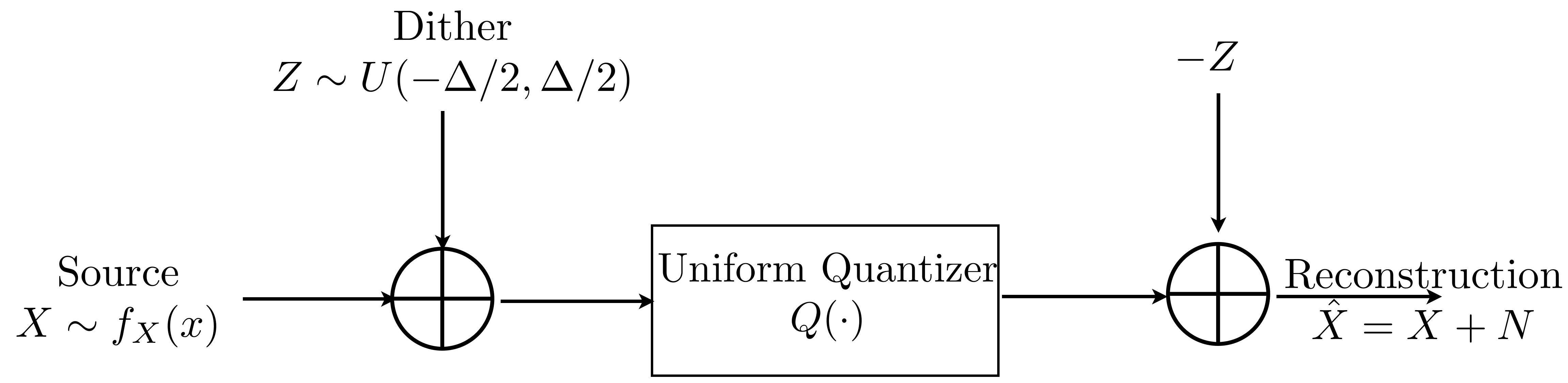} \centering \caption{The basic structure of dithered quantization}
\label{fig1} %
\end{figure*}

Beyond its theoretical significance, randomized quantization is of
practical interest. Many filter/system optimization problems in practical
compression settings, such as the ``rate-distortion optimal filterbank
design'' problem \cite{mihcak}, or low rate filter optimization for
DPCM compression of Gaussian auto-regressive processes \cite{guleryuz},
assume quantization noise that is independent of (or uncorrrelated
with) the source. Although this assumption is satisfied at asymptotically
high rates \cite{gershobook}, such systems are mostly useful for
very low rate applications. For example, in \cite{guleryuz}, it is
stated that the assumptions made in the paper are not satisfied by
deterministic quantizers, and that dithered quantizers satisfy the
assumptions exactly. However, conventional (uniform) dithered quantization
suffers from suboptimal compression performance. Hence, a quantizer
that mostly satisfies the assumptions, but at  minimal cost in
performance degradation, would have considerable impact on many such
applications. 

In this paper, we consider a generalization to enable
effective dithering of nonuniform quantizers. To the best of our knowledge,
this paper is the first attempt (other than our preliminary work in
\cite{akyol2009nonuniform,akyoltowards}) to consider dithered quantization in
a nonuniform quantization framework. One immediate problem with nonuniform
dithered quantization is how to apply dithering to unequal quantization
intervals. In traditional dithered quantization, the dither signal
is matched to the uniform quantization interval while maintaining
independence of the source, but it is not clear how to match the generic
dither to varying quantization intervals. As a remedy to this problem,
we propose dithering in the companded domain. We derive the closed form necessary conditions
for optimality of the compressor and expander mappings 
for both fixed and variable rate randomized quantization. We numerically optimize the 
mappings by iteratively imposing these necessary conditions.

 However,  the resulting (unconstrained randomized)
quantizer does not render reconstruction error orthogonal to the source. Therefore,
we extend the framework to include an explicit such constraint. Surprising
theoretical results show direct and simple connections between the
optimally constrained random quantizers and their unconstrained counterparts.
We note in passing that the nonuniform dithered quantizer subsumes
the conventional uniform dithered quantizer as an extreme special
case. 

For the variable rate case, the proposed nonuniform dithered quantizer is expected to outperform
the conventional dithered quantizer, most significantly at low rates
where the optimal variable rate (entropy coded) quantizer is often
far from uniform. We observe that a deterministic quantizer cannot
render the quantization noise independent of the source but can make
it uncorrelated with the source. We hence also present an alternative
deterministic quantizer that provides quantization noise uncorrelated
with the source. We derive the optimality conditions of such constrained
quantizers, for both fixed and variable rate quantization, and compare
their rate-distortion performance to that of randomized quantizers. 

Dithered quantization offers an interesting theoretical twist. Randomized
quantization is an instance of the random encoding principle used
to elegantly prove the achievability of coding bounds in rate distortion
theory \cite{coverbook}. However, to actually achieve those bounds,
a random encoding scheme is not necessary, as they can be approached
by a sequence of deterministic quantizers of increasing block length.
In the second part of the paper, we investigate the settings under
which randomized quantization is asymptotically necessary. A trivial
example involves requiring source-independent quantization error.
It is obvious that the reconstruction (hence quantization error) is
a deterministic function of the source when the quantizer is deterministic
\cite{gershobook}, while conventional dithered quantization produces
quantization error that is independent of the source. Although a deterministic
quantizer can never render the quantization error independent of the
source, it can produce quantization error uncorrelated with the source. A natural
question is whether the rate distortion bound, subject to the uncorrelated
error constraint, can be achieved (asymptotically) with a deterministic
quantizer.

The paper is organized as follows: In Section III, we present the proposed nonuniform randomized quantizers, along with its extension to constrained randomized quantizer that renders the quantization error orthogonal to the source. In Section IV, we derive the necessary conditions of optimality for the deterministic quantizer that generates
reconstruction error uncorrelated with the source. In Section V, we
study the asymptotic (in quantizer dimension) results, and show that for a Gaussian source optimal constrained quantizer must be randomized. Experimental
results that compare the proposed quantizers to the conventional dithered quantizer are presented in Section VI. We discuss the obtained results and summarize the contributions in Section VII.

\begin{figure*}
\centering \includegraphics[scale=0.28]{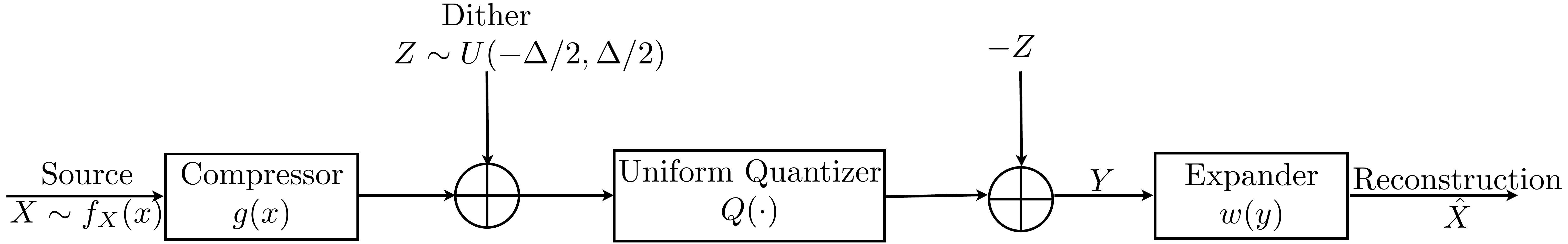} \caption{The proposed nonuniform dithered quantizer}

\label{fig2} %
\end{figure*}

\section{Review of Dithered Quantization}

\subsection{Notation and Preliminaries}
 In general, lowercase letters (e.g., $x$) denote scalars, boldface lowercase (e.g., $\boldsymbol x$) vectors, uppercase (e.g., $U, X$) matrices and random variables, and boldface uppercase (e.g., $\boldsymbol X$) random
vectors.  $\mathbb E[\cdot]$,  $ R_X$, and $  R_{XZ}$ denote the expectation, covariance of $\boldsymbol X $ and cross covariance of $\boldsymbol X$ and $\boldsymbol Z$ respectively\footnote{We assume zero mean random variables. This assumption is not necessary, but it considerably simplifies the notation. Therefore, it is kept throughout the paper.}. $\nabla$ denotes the gradient.  $\cal N(\boldsymbol \mu, \boldsymbol K)$ denotes the Gaussian random vector with mean $\boldsymbol \mu$ and covariance matrix $ K$.

The entropy of a discrete random vector source $\boldsymbol{X}\in\mathbb{R}^{K}$
taking values in ${\cal X}$ is

\begin{equation}
H(\boldsymbol{X})=-\sum_{\boldsymbol{x}\in{\cal X}}P(\boldsymbol{X}=\boldsymbol{x})\log P(\boldsymbol{X}=\boldsymbol{x})\end{equation}
 where logarithm is base 2 to measure it in bits. The differential
entropy of a continuous random variable $\boldsymbol{X}$ with probability
density function $f_{X}(x)$ is \begin{equation}
h(\boldsymbol{X})=-\int f_{X}(\boldsymbol{x})\log f_{x}(\boldsymbol{x})d\boldsymbol{x}\label{difent}\end{equation}
 The divergence between two densities $f_{X}$ and
$g_{X}$, is given by \begin{equation}
\mathcal{D}(f_{X}||g_{X})=\int f_{X}(\boldsymbol{x})\log\frac{f_{X}(\boldsymbol{x})}{g_{X}(\boldsymbol{x})}d\boldsymbol{x}\end{equation}
The divergence definition above can be extended to conditional densities. For joint densities, $f_{XY}$ and $g_{XY}$ the conditional divergence $\mathcal{D}(f_{X|Y}||g_{X|Y})$  is defined as the divergence between the conditional distributions $f_{X|Y}$ and $g_{X|Y}$ averaged over the density $f_Y(\boldsymbol y)$: 
\begin{equation}
\mathcal{D}(f_{X|Y}||g_{X|Y})\!=\! \int \!  f_{Y}\! (\boldsymbol{y}) \! \int \! f_{X|Y}(\boldsymbol{x,y})\log\frac{f_{X|Y}(\boldsymbol{x,y})}{g_{X|Y}(\boldsymbol{x,y})}d\boldsymbol{x}d\boldsymbol{y}
\end{equation}

 The mutual information between two random variables $\boldsymbol X$ and
$\boldsymbol Y$ with marginal densities $f_X(\boldsymbol x)$ and $f_Y(\boldsymbol y)$ and a joint density $f_{XY}(\boldsymbol x,\boldsymbol y)$  is given by 
\begin{equation}
{I}(\boldsymbol X,\boldsymbol Y)=\int \int f_{X, Y}(\boldsymbol{x},\boldsymbol{y})\log\frac{f_{X,Y}(\boldsymbol{x,y})}{f_{X}(\boldsymbol{x})f_{Y}(\boldsymbol{y})}d\boldsymbol{x}d\boldsymbol{y}
\end{equation}

Zero-mean vectors $\boldsymbol{x}\in \mathbb R^K$ and $\boldsymbol{y}\in \mathbb R^M$ are said
to be uncorrelated if they are orthogonal: \begin{equation}
\mathbb{E}[ \boldsymbol{y}\boldsymbol{x}^{T} ]={0}\label{rate}\end{equation}
where the right hand size is $M\times K$ matrix of zeros. 


\subsection{Dithered Quantization}
 A quantizer is defined by a set of reconstruction
points and a partition. The partition $\mathcal{P}=\{\mathcal{P}_{i}\}$
associated with a quantizer is a collection of disjoint regions whose
union covers $\mathbb{R}^{K}$. The reconstruction points $\mathcal{R}=\{\boldsymbol{r}_{i}\}$
are typically chosen to minimize a distortion measure. The vector
quantizer is a mapping $\mathcal{Q}_{K}: \mathbb{R}^{K}\rightarrow\mathbb{R}^{K}$ that 
maps every vector $\boldsymbol{X}\in\mathbb{R}^{K}$ into the reconstruction
point that is associated with the cell containing $\boldsymbol{X}$,
i.e. \begin{equation}
Q_{K}(\boldsymbol{X})=\boldsymbol{r_{i}}\,\,\,\text{if }\,\,\,\boldsymbol{X}\in\mathcal{P}_{i}\end{equation}

While our theoretical results are general, for a vector quantizer
of arbitrary dimensions, for presentation simplicity, we will primarily
focus on scalar quantization in the treatment of numerical optimization
of nonuniform dithered quantizer and for experimental results. The
nonuniform dithered quantization approach is directly extendable to
vector quantization by replacing the companded domain uniform quantizer with a lattice
quantizer, although at the cost of significantly more challenging
numerical optimization. 

The scalar uniform quantizer, with reconstructions
\{$0,\pm\Delta,\pm2\Delta,...,\pm T\Delta$\}, is a mapping $Q:\mathbb{R}\rightarrow\mathbb{R}$
such that \begin{equation}
Q(x)=i\Delta\,\,\,\text{ for }\,\, i\Delta-\Delta/2<x \leq i\Delta+\Delta/2\end{equation}
 In fixed rate quantization, the range parameter $T$ is determined
by the rate $R_{f}$ \begin{equation}
R_{f}=\log(2T+1)\end{equation}
 while in variable rate quantization $T$ need not, in principle,
be finite and we will assume $T\rightarrow\infty$. In this case,
uniform quantization is followed by lossless source encoding (entropy
coder). 

 Let dither $Z$ be a random variable,
distributed uniformly on the interval $(-\Delta/2,\Delta/2)$. Then,
conventional dithered quantizer approximates the source $X$ by \begin{equation}
\hat {X}=Q(X+Z)-Z\end{equation}

 It can be shown that the reconstruction error of this quantizer (denoted $N$)
is independent of the source value $X=x$, i.e., $N=\hat{X}-X=Q(X+Z)-Z-X$
 is independent of $X$ and uniformly distributed
over $(-\Delta/2,\Delta/2)$ for all $X$.  Contrast that with a deterministic quantizer, whose error is completely
determined by the source value \cite{gershobook}.

We note that for this property to hold, the quantizer should span the support of the source density i.e., there should be no overload distortion. While this is often the case for variable rate quantization, for fixed rate overload distortion is inevitable if the source has unbounded support such as a Gaussian source. For practical purposes though, it is common to assume that the source has finite support and we also follow this assumption in our analysis of fixed rate randomized quantization: the quantization error of conventional (uniform) dithered quantization is assumed to be independent of the source.

The realization of the dither random variable $Z$ is available to
both the encoder and the decoder. Thus, assuming an optimal entropy coder, the rate of the variable rate quantizer
tend to the conditional entropy of the reconstruction given the dither,
i.e., \begin{equation}
R_{v}=H(\hat{X}|Z)=H(Q(X+Z)|Z)\end{equation}
 In \cite{zamiruqr}, it was shown that the following holds: \begin{equation}
H(Q(X+Z)|Z)=h(X+N)-\log\Delta\label{rate}\end{equation}

We will use (\ref{rate}) in the rate calculations of the variable rate (entropy coded) randomized quantization. 


\section{Nonuniform Dithered Quantizer}

The main idea is to circumvent the main difficulty due to unequal
quantization intervals by performing uniform dithered quantization
in the companded domain (see Figure 2). The source $X$
is transformed through compressor $g(\cdot)$ before undergoing dithered uniform
quantization. At the decoder side, the dither is subtracted to obtain
$Y$. Since we perform uniform dithered quantization in the companded
domain, it is easy to show that $Y=g(X)+N$, where $N$ is uniformly
distributed over $(-\Delta/2,\Delta/2)$ and independent of the source.
The reconstruction is obtained by applying the expander $\hat{X}=w(Y)$.
The objective is to find the optimal compressor and expander mappings
$g(\cdot),w(\cdot)$ that minimize the expected distortion under the rate
constraint. The MSE distortion can be written as: \begin{equation}
D=\int\int[x-w(g(x)+n)]^{2}f_{X}(x)f_{N}(n)dxdn\label{dist-1}\end{equation}
 where $f_{N}(n)$ is uniform over $(-\Delta/2,\Delta/2)$. Interestingly,
this problem bears some similarity to the joint source channel mapping
problem where the optimal analog encoding and decoding mappings are
studied \cite{emrah_itw10}. In our setting, the quantization error
is analogous to the channel noise and the rate constraint in variable rate quantization plays a
role similar to that of the power constraint. Similar to \cite{emrah_itw10},
we develop an iterative procedure that enforces the necessary conditions
for optimality of the mappings. Note that the conventional (uniform) dithered
quantizer is a special case employing the trivial identity mappings,
i.e., $g(x)\!=\! w(x)\!=\! x,\forall x$.


\subsection{Optimal Expander}

\noindent The conditional expectation $w(y)=\mathbb{E}\{X|Y=y\}$ minimizes
MSE between the source and the estimate. Hence, the optimal
expander $w$ is
 \begin{equation}
w(y)=\frac{\int\limits _{\gamma_{-}}^{\gamma_{+}}xf_{X}(x)dx}{\int\limits _{\gamma_{-}}^{\gamma_{+}}f_{X}(x)dx}\label{expander}\end{equation}
 where for fixed rate $\gamma_{+}=\min\{g^{-1}(\Delta T),g^{-1}(y+\Delta/2)\}$
and $\gamma_{-}=\max\{g^{-1}(-\Delta T),g^{-1}(y-\Delta/2)\}$, while
for variable rate $\gamma_{+}=g^{-1}(y+\Delta/2)$ and $\gamma_{-}=g^{-1}(y-\Delta/2)$. \\

 {\bf Note}: We restrict the discussion to regular quantizers throughout this paper, hence $g(\cdot)$ is monotonically increasing.
\subsection{Optimal Compressor}
Unlike the expander, the optimal compressor cannot be written in closed
form. However, a necessary optimality condition can be obtained by
setting the functional derivative of the cost to zero. Thus, a locally optimal compressor $g(\cdot)$,  for a given expander $w(\cdot)$, requires that the functional derivative of the total cost,
$J$, along the direction of any variation function $\eta(\cdot)$ vanishes
\cite{Luenberger}, i.e., \begin{equation}
\nabla J=\frac{\partial}{\partial\epsilon}\bigg|_{\epsilon=0}J\left[{g}({x})+\epsilon{\eta}({x})\right]=0,\,\forall x\in\mathbb{R}\end{equation}
for all admissible perturbation functions $\eta(\cdot)$.
\subsubsection{Fixed rate}
For fixed rate, we have granular distortion, denoted $D_{g}$, and
overload distortion, denoted $D_{ol}$. Note that we must account for the overload
distortion here, as this constrains $g(x)$ from growing unboundedly in the iterations of the proposed algorithm. Since the rate is fixed, the total
cost is identical to the distortion, i.e., $J_{f}=D_{g}+D_{ol}$ where $D_{g}$ and $D_{ol}$ are:
 \begin{equation}
D_{g}\!=\!\frac{1}{\Delta}\int\limits _{-\Delta/2}^{\Delta/2}\int\limits _{g^{-1}(-\Delta T)}^{g^{-1}(\Delta T)}[x-w(g(x)+n)]^{2}f_{X}(x)dxdn\label{dist_fixed}
\end{equation}
 \begin{align}
D_{ol}\! &=  \frac{1}{\Delta}\int\limits _{-\Delta/2}^{\Delta/2}\left \{ \int\limits _{-\infty}^{g^{-1}(-\Delta T)}[x-w(-T\Delta+n)]^{2}f_{X}(x)dx  \right .\nonumber \\
 &\left .+  \int\limits _{g^{-1}(\Delta T)}^{\infty}[x-w(T\Delta+n)]^{2}f_{X}(x)dx \right \}dn
\end{align}

\subsubsection{Variable rate}
 The rate is obtained via (12) and (\ref{rate}), which require the distribution of
$Y=g(X)+N$: 
 \begin{equation}
f_{Y}(y)=\frac{1}{\Delta}\left[F_{X}(g^{-1}(y+\Delta/2))\!-\! F_{X}(g^{-1}(y-\Delta/2))\right]\end{equation}
 where $F_{X}(x)$ is the cumulative distribution function of $X$.
The rate is then evaluated as \begin{equation}
R_{v}=h(Y)-\log\Delta\end{equation}
 The total cost for variable rate quantization is $J_{v}=D+\lambda R$
 where $\lambda$ is the Lagrangian parameter that is adjusted to
obtain the desired rate.

\subsection{Design Algorithm}

The basic idea is to iteratively alternate between enforcing the individual necessary
conditions for optimality, thereby successively decreasing the total
cost. Iterations are performed until the algorithm reaches a stationary
point. Solving for the optimal expander is straightforward since the
expander is expressed in closed form as a functional of the known
quantities, ${g}({\cdot})$, $f_{X}({\cdot})$. Since the compressor condition
is not in closed form, we perform steepest descent, i.e., move in
the direction of the functional derivative of the total cost with
respect to the compressor mapping $g$. By design, the total cost
decreases monotonically as the algorithm proceeds iteratively. The
compressor mapping is updated according to (\ref{g_iter}), where
$i$ is the iteration index, $\nabla{J[g(\cdot)]}$ is the directional derivative
and $\mu$ is the step size. 
\begin{equation}
{g}_{i+1}({x})={g}_{i}({x})-\mu\nabla{J[g]}\label{g_iter}
\end{equation}

Note that there is no guarantee that an iterative descent algorithm
of this type will converge to the globally optimal solution. The algorithm
will converge to a local minimum and hence, initial conditions are important in such greedy optimizations. A low
complexity approach to mitigate the poor local minima problem, is
to embed within the solution {}the ``noisy channel relaxation''
method of \cite{gadkari1999robust,Knagenhjelm}. We initialize the
compressor mapping with random initial conditions and run the algorithm
for a very low rate (large value for the Lagrangian parameter $\lambda$). Then,
we gradually increase the rate (decrease $\lambda$) while tracking the
minimum. Note that local minima problem is more pronounced at multi-dimensional optimizations, which hence requires more powerful non-convex optimization tools such as deterministic annealing \cite{da}. In our design and experiments, we focus on scalar compressor and expander and  we did not observe significant local minima problems.  

\section{Reconstruction Error Uncorrelated with the Source }
 In this section, we propose two quantization schemes (one deterministic, one randomized) that satisfy the constraint that  reconstruction error be uncorrelated with the source. 
\subsection{Constrained Deterministic  Quantizer}
A deterministic quantizer cannot yield quantization noise independent
of the source \cite{gershobook}. However, it is possible to render
the quantization noise uncorrelated with the source. An early  prior work
along this line appeared in \cite{oba}, where a uniform quantizer is
converted to a quantizer whose quantization noise is uncorrelated
with the source, by adjusting the reconstruction points. In this section, we derive the optimal (nonuniform in general) deterministic
quantizer which is constrained to give quantization error uncorrelated
with the source. 

Let $\boldsymbol{r}_{i}$ and ${\widehat{\boldsymbol{r_{i}}}}$ be
the reconstruction points and $\mathcal{P}_{i}$ and ${\widehat{\mathcal{P}_{i}}}$
represent the $i^{th}$ quantization region, for the constrained (i.e.,
whose quantization error is uncorrelated with the source) and unconstrained
MSE optimal quantizer, respectively. Also, let $p_{i}$ and $\widehat{p}_{i}$
denote the probability of $\boldsymbol X$ falling into the $i^{th}$ cell of
these respective quantizers.
 \begin{theorem}$\mathcal{P}_{i}={\widehat{\mathcal{P}_{i}}}\,\,\,\text{and}\,\,\,\boldsymbol{r_{i}}=\boldsymbol{C}{\boldsymbol{\hat{r}_{i}}},\forall i$
where
 $$\boldsymbol{C}={R}_{X}\left(\sum\limits _{i=1}^{M}p_{i}{\boldsymbol{\hat{r}_{i}}}{\boldsymbol{\hat{r}_{i}^{T}}}\right)^{-1}.$$
 \end{theorem}
\begin{proof}
We start with the fixed rate analysis. Let $M$ denote the number
of quantization cells. The distortion can be expressed as
\begin{equation}
D=\sum\limits _{i=1}^{M}\int\limits _{\boldsymbol{x}\in\mathcal{P}_{i}}{(\boldsymbol{x}-\boldsymbol{r}_{i})^{T}(\boldsymbol{x}-\boldsymbol{r}_{i})f_{X}(\boldsymbol{x})\, d\boldsymbol{x}}\label{distortion}\end{equation}
 and the {}``uncorrelatedness'' constraint may be stated
via the orthogonality principle \begin{equation}
\sum\limits _{i=1}^{M}\int\limits _{\boldsymbol{x}\in\mathcal{P}_{i}}\boldsymbol{x}(\boldsymbol{x}-\boldsymbol{r}_{i})^{T}f_{X}(\boldsymbol{x})\, d\boldsymbol{x}=\boldsymbol{0}\label{constraint}\end{equation}
 Note further that (\ref{constraint}) can be written as: \begin{equation}
\sum\limits _{i=1}^{M}\boldsymbol{r}_{i}\boldsymbol{l}_{i}^{T}={R}_{X}\text{ where }\boldsymbol{l}_{i}=\int\limits _{\boldsymbol{x}\in\mathcal{P}_{i}}\boldsymbol{x}f_{X}(\boldsymbol{x})\, d\boldsymbol{x}\label{constraint2}\end{equation}
 The constrained problem of minimizing $D$ subject to $\sum\limits _{i=1}^{M}\boldsymbol{r}_{i}\boldsymbol{l}_{i}^{T}={R}_{X}$
is equivalent to the unconstrained minimization of Lagrangian $J$:
\begin{equation}
J=D+\sum\limits _{k=1}^{K}\boldsymbol{\gamma}(k)^{T}\left[{R}_{X}(k)-\sum\limits _{i=1}^{M}\boldsymbol{r}_{i}\boldsymbol{l}_{i}(k)\right]\label{lagcost}
\end{equation}
 where $\boldsymbol{\gamma}=[\boldsymbol{\gamma}(1)\,\,\,\boldsymbol{\gamma}(2)...\,,\boldsymbol{\gamma}(K)]$
denotes the $K\times K$ multiplier Lagrangian matrix, ${R}_{X}(k)$
denotes the $k^{th}$ column  of ${R}_{X}$ and $l_{i}(k)$
denotes the $k^{th}$ element of $\boldsymbol{l}_{i}$. 
By setting $\nabla_{r_{i}}J=\boldsymbol{0}$, we obtain the condition: 
\begin{equation}
\nabla_{r_{i}}J=-2\boldsymbol{l}_{i}^{T}+2p_{i}\boldsymbol{r}_{i}^{T}-\sum\limits _{k=1}^{K}\gamma(k)^{T}{l}_{i}(k)=\boldsymbol{0}\label{grad}
\end{equation}
 Noting that $\sum_{k=1}^{K}\gamma(k)^{T}{l}_{i}(k)=\boldsymbol{l}_{i}^{T}\boldsymbol{\gamma}$,we obtain $\boldsymbol{r}_{i}=\frac{1}{p_{i}}{C}\boldsymbol{l}_{i}\label{constant}$ where ${C}$ is a constant $M\times M$ matrix. $\boldsymbol{C}$ is found by plugging this into (\ref{constraint2}):
\begin{equation}
\boldsymbol{C}={R}_{X}\left(\sum\limits _{i=1}^{M}\frac{1}{p_{i}}\boldsymbol{l}_{i}\boldsymbol{l}_{i}^{T}\right)^{-1}\label{c}
\end{equation}
 Note that $\boldsymbol{l}_{i}/p_{i}$ is the MSE optimal reconstruction
of an unconstrained quantizer that shares the same decision boundary with the constrained one, 
$\mathcal{P}_{i}$. Plugging (\ref{c}) into (\ref{distortion}) and
after some algebraic manipulations, we obtain: \begin{equation}
D=\frac{\sigma_{X}^{2}}{\sigma_{X}^{2}-D^{*}}D^{*}\label{dist}\end{equation}
 where $D^{*}$ is the distortion associated with the quantizer given
by $\mathcal{P}_{i}$ and with corresponding optimal reconstruction
points $\boldsymbol{l}_{i}/p_{i}$. (\ref{dist}) implies that $D$
achieves its minimum whenever $D^{*}$ is minimized. Hence, \begin{equation}
\mathcal{P}_{i}=\widehat{\mathcal{P}_{i}}\label{pi} \Rightarrow \boldsymbol l_i=p_i \boldsymbol{\hat{r_i}}
\end{equation}
Plugging (\ref{pi}) into  (\ref{c}), we obtain the result.  The proof for variable rate follows similar lines, with the only
modification that we now have to account for the rate term $R=\sum\limits _{i=1}^{M}p_{i}\log p_{i}$.
The uncorrelatedness constraint is identical to the one in fixed rate,
hence the overall Lagrangian cost can be expressed as: \begin{equation}
J=D+\sum\limits _{k=1}^{M}\boldsymbol{\gamma}(k)^{T}\left[\boldsymbol{R}_{X}(k)-\sum\limits _{i=1}^{M}\boldsymbol{r}_{i}\boldsymbol{l}_{i}(k)\right]+\lambda\sum\limits _{i=1}^{M}p_{i}\log p_{i}\end{equation}
 By setting $\nabla_{r_{i}}J=\boldsymbol{0}$ and following the same
steps, we obtain: \begin{equation}
\boldsymbol{r}_{i}=\frac{1}{p_{i}}\boldsymbol{R}_{X}\left(\sum\limits _{i=1}^{M}\frac{1}{p_{i}}\boldsymbol{l}_{i}\boldsymbol{l}_{i}^{T}\right)^{-1}\boldsymbol{l}_{i}\label{constant2}\end{equation}
 Note that (\ref{dist}) holds due to (\ref{constant2}) and the optimal
unconstrained quantizer achieves the minimum distortion $D^{*}$ subject
to the rate constraint. This indicates that the constrained and unconstrained
quantizers have identical $p_{i}$ and hence $\mathcal{P}_{i}=\widehat{\mathcal{P}_{i}}$ which implies $\boldsymbol l_i=p_i \boldsymbol{\hat{r_i}}$.
Plugging $\boldsymbol l_i=p_i \boldsymbol{\hat{r_i}}$ into (\ref{constant2}), we obtain the desired result.

\end{proof}

\subsection{Constrained Randomized Quantizer}
Due to the effect of companding, the nonuniform randomized quantizer we derived in Section III does not guarantee reconstruction error uncorrelated
with the source even though it builds on the (conventional) dithered
quantizer whose quantization error is independent of the source. We
therefore explicitly constrain the randomized quantizer to generate
uncorrelated reconstruction error, by adding a penalty term to the
total cost function. The Lagrangian parameter $\lambda_{c}\geq0$
is set to ensure $\mathbb{E}\{xw(g(x)+n)\}=\mathbb{E}\{x^{2}\}$. 
\begin{equation}
J_{c}=J+\lambda_{c}\mathbb{E}[x^2-x w(g(x)+n)] 
\label{top}
\end{equation}
 where $J=J_{v}$ in the case of variable rate and $J=J_{f}$ for
fixed rate. We find the necessary conditions of optimality of constrained compressor and expander mappings at  fixed
and variable rate, by setting the functional derivative of the total cost ($J_c$) to zero. Surprisingly,
the optimally constrained compressor mapping remains unchanged (compared to the unconstrained optimal compressor) and the only modification
of the optimally constrained expander mapping is simple scaling. We state this result
in the following theorem.
\begin{theorem} Let $g$ and $w$ be the  compressor and expander
mappings of the unconstrained optimal randomized quantizer. Let $g_{c}$
and $w_{c}$ denote the optimal mappings subject to the constraint
that the reconstruction error be uncorrelated with the source. Then,
\begin{equation}
g_{c}(x)=g(x),w_{c}(y)=(1-\lambda_{c})w(y)\end{equation}
where $\lambda_c$ is the Lagrangian multiplier of (\ref{top}).
 \end{theorem}
 Note that this result applies to both fixed and variable rate. 
\begin{proof} The optimal expander is no longer the standard conditional
expectation, since it is impacted by the constraint. By setting 
\begin{equation}
\frac{\partial}{\partial\epsilon}\bigg|_{\epsilon=0}J_{c}\left[{w}({y})+\epsilon{\eta}({y})\right]=0
\end{equation}
we obtain the optimal expander in closed form as $w_{c}(y)=(1-\lambda_{c})w(y)$.
The update rule for $g_{c}(x)$ can be derived similarly. Setting 
\begin{equation}
\frac{\partial}{\partial\epsilon}\bigg|_{\epsilon=0}J_{c}\left[{g}({x})+\epsilon{\eta}({x})\right]=0
\end{equation}
and plugging $w_{c}(y)=(1-\lambda_{c})w(y)$ yields, after straightforward
algebra, $g_{c}(x)=g(x)$.
\end{proof}


\section{Asymptotic Analysis}

\subsection{Rate-Distortion Functions}

To quantify theoretically how much a source%
\footnote{The notation in this section is limited to scalar sources for simplicity,
it is trivial to extend the results to vector sources albeit with
more complicated notation. %
} can be compressed under the independent/uncorrelated reconstruction
error constraint, we define two rate-distortion functions in which
we respectively constrain the reconstructions error to be i) uncorrelated
with the source $R_{U}(D)$, and ii) independent of the source $R_{I}(D)$.

Assume that we have source $X $ with density $f_X(\cdot)$ that produces the independent identically distributed (i.i.d.) sequence $X_1, X_2, .., X_n$ denoted as $X^n$. Similarly, let ${\hat X_1}, {\hat X_2}, .., {\hat X_n}$ be the reconstruction sequence, denoted as ${\hat X^n}$. Let  $S^n=X^n-{\hat X^n}$ be the i.i.d. sequence of reconstruction errors with marginal density $ f_S(\cdot)$. Let  $f_{XS}(x,s)$ denote
 joint distribution of $X$ and $S$ and $d(X^{n},\hat {X}^n)$
denote the  distortion measure between sequences $X^{n}$
and $\hat{X}^n$ defined as 
\begin{equation}
d(X^{n},\hat {X}^n)=\frac{1}{n}\sum\limits_{i=1}^{n}d(X_i,\hat {X}_i)
\end{equation}
Let us recall the classical rate-distortion result in information theory: 

  \begin{unnumbered}[eg. \cite{coverbook}]
  Let $R(D)$ be the infimum of all achievable rates $R$ with average distortion $\mathbb{E}[d(X^{n},X^n+S^{n})] \leq D$ as $n \rightarrow \infty$. Then, 
  \begin{equation}
R(D)=\inf_{\substack{S:\mathbb{E}[d(X,X+S)]\leq{D} }} I(X;X+S)\label{Runc}\end{equation}
 \end{unnumbered}

We next focus on our problem: let $R_{U}(D)$ be the infimum of all achievable
rates $R$ with average distortion 
\begin{equation}
\lim_{n \rightarrow \infty} \mathbb{E}[d(X^{n},X^n+S^{n})] \leq D
\end{equation}
subject to the constraint 
\begin{equation}
\mathbb E [X_{i}S_{i}]=0,  i=1,2,.., n
\end{equation}
as ${n\to\infty}$. Similarly, let $R_{I}(D)$ be the infimum of all achievable rates $R$ with average
distortion $\mathbb{E}[d(X^{n},X^n+S^{n})] \leq D$ subject to the constraint
$S_i$ is independent of $X_i$ for all $i$, as ${n\to\infty}$. Then, we have the following result characterizing the fundamental limits of source compression under the constraints that reconstruction error is uncorrelated with or independent of the source.

\begin{theorem}\begin{equation}
R_{U}(D)=\inf_{\substack{S:\mathbb{E}[d(X,X+S)]\leq{D}\\
\mathbb{E}[XS]=0}
}I(X;X+S)\label{Runc}\end{equation}
 \begin{equation}
R_{I}(D)=\inf_{\substack{S:{\mathbb{E}[d(X,X+S)]}\leq{D}\\
\mathcal{D}(f_{XS}(X,S)||f_{X}(X)f_{S}(S))=0}
}I(X;X+S)\label{Rind}\end{equation}
 \end{theorem}
 
  \begin{proof} Consider the distortion measures
  \begin{equation}
d_{U}(x,x+s)=d(x,x+s)+\beta ||xs||\label{unc}\end{equation}
 \begin{equation}
d_{I}(x,x+s)=d(x,x+s)+\beta\log \frac{f_{XS}(x,s)}{f_{X}(x)f_{S}(s)}\label{ind}\end{equation}
for some $\beta >0$. 

We next consider the rate-distortion functions (denoted $R_{U}^*(D)$ and $R_{I}^*(D)$) associated with these distortion measures. By replacing $d$ with $d_{U}$
and $d_{I}$ in the standard rate-distortion functions, we obtain the following expressions: 
\begin{equation}
R_{U}^*(D)=\inf_{\substack{S:\mathbb{E}[d_{U}(X,X+S)]\leq D}
}I(X;X+S)\end{equation}
 \begin{equation}
R_{I}^*(D)=\inf_{\substack{S:\mathbb{E}[d_{I}(X,X+S)]\leq D}
}I(X;X+S)\end{equation}

We note that the achievability and the converse proofs are straightforward extensions
of the standard achievability and the converse proofs for regular
rate distortion function.

We next consider the distortion measures  $d_{U}$ and $d_{I}$ and associated rate-distortion functions $R_{U}^*(D)$ and $R_{I}^*(D)$ when $\beta \rightarrow \infty$. As $\beta \rightarrow \infty$,  $\mathbb E [d_{U}(X^n,X^n+S^n)] \leq D$ implies $\mathbb E [d(X^n,X^n+S^n)] \leq D$  while $\mathbb E [X_{i}S_{i}] \rightarrow 0$ for all $i$.  Similarly, as $\beta \rightarrow \infty$, $\mathbb E [d_{I}(X^n,X^n+S^n)] \leq D$  implies  $\mathbb E [d(X^n,X^n+S^n)] \leq D$ under the constraint that $X_i$ and $S_i$ are asymptotically independent for all $i$. Hence, as  $\beta \rightarrow \infty$, the distortion measures under consideration satisfy the respective requirements of uncorrelatedness or independence, i.e., $R_{U}(D)= R_{U}^*(D)$ and  $R_{I}(D)= R_{I}^*(D)$.

Hence, (\ref{Runc}) and (\ref{Rind}) are indeed the information
theoretical characterization of the limits of encoding a source with
uncorrelated and independent reconstruction error respectively. 
\end{proof}


\subsection{Gaussian Vector Source with MSE Distortion}

In this section, we examine a special case where the source is vector
Gaussian and the distortion measure is squared error. We start with
an auxiliary lemma without proof (see eg. \cite{coverbook} for the
proof). 

\begin{lemma}[\cite{coverbook}] \label{equality_lemma}

Let $\boldsymbol{S} \sim f_{S}$ and $\boldsymbol{S}_{G} \sim f_{S_{G}}$ be random vectors
in $\mathbb{R}^{K}$with the same covariance matrix ${R}_{S}$.
If $\boldsymbol{S}_{G}\sim\mathcal{N}(\boldsymbol{0},{R}_{S})$ then
 \begin{equation}
\mathbb{E}_{S_{G}}[\log(f_{S_{G}}(\boldsymbol{S}))]=\mathbb{E}_{S}[\log(f_{S_{G}}(\boldsymbol{S}))]
\end{equation}
where  $\mathbb{E}_{S_{G}}$ and $\mathbb{E}_{S}$ denote the expectations
with respect to $f_{S_G}$ and $f_S$ respectively.
\end{lemma}

Let us present a key lemma regarding the mutual information of two
correlated random vectors constrained to have a fixed cross covariance
matrix. 

\begin{lemma}
\label{main_lemma}
Let $\boldsymbol{X}\sim\mathcal{N}(\boldsymbol{0},{R}_{X})$
and $\boldsymbol{S}_{G}\sim\mathcal{N}(\boldsymbol{0},{R}_{S})$ be jointly Gaussian random vectors
in $\mathbb{R}^{K}$. Let $\boldsymbol{S}\in \mathbb{R}^{K}$ and $ \boldsymbol{S}_{G}$ have the same covariance matrix, ${R}_{S}$
and the same cross covariance matrix with $\boldsymbol{X}$, ${R}_{SX}$. Then, 
\begin{equation}
I(\boldsymbol{X},\boldsymbol{X}+\boldsymbol{S})\geq I(\boldsymbol{X},\boldsymbol{X}+\boldsymbol{S}_{G})
\end{equation}
 with equality if and only if $\boldsymbol S\sim\mathcal{N}(\boldsymbol 0,{R}_{S})$.
\end{lemma}

 \begin{proof} 
 Consider $\gamma\!=\! I(\boldsymbol{X},\boldsymbol{X}\!+\!\boldsymbol{S})\!-\! I(\boldsymbol{X},\boldsymbol{X}\!+\!\boldsymbol{S}_{G})$. Plugging the expressions, we obtain:
\begin{equation}
\gamma= h(\boldsymbol{X}|\boldsymbol{S_G}+\boldsymbol{X})- h(\boldsymbol{X}|\boldsymbol{S}+\boldsymbol{X})
\end{equation}
Noting that $h(\boldsymbol{X}|\boldsymbol{S_G}+\boldsymbol{X})=h(\boldsymbol{S_G}|\boldsymbol{S_G}+\boldsymbol{X})$ and $h(\boldsymbol{X}|\boldsymbol{S}+\boldsymbol{X})=h(\boldsymbol{S}|\boldsymbol{S}+\boldsymbol{X})$ and plugging $\boldsymbol{Y}=\boldsymbol{X}+\boldsymbol{S}$ and $\boldsymbol{Y}_{G}=\boldsymbol{X}+\boldsymbol{S}_{G}$, we obtain: 
\begin{align}
 & \gamma=h(\boldsymbol{S}_{G}|\boldsymbol{Y}_{G})-h(\boldsymbol{S}|\boldsymbol{Y})=\\
 &\!\!\!\int\!\!\!\int\!\!\! \left \{ f_{S,Y}(\boldsymbol{s},\!\boldsymbol{y})\!\log f_{S|Y}(\boldsymbol{s},\boldsymbol{y})\!\!\! -\!\!\! f_{S_{G},Y_{G}}(\boldsymbol{s},\!\boldsymbol{y})\!\log f_{S_{G}|Y_{G}}\!(\boldsymbol{s},\boldsymbol{y})\!\right \}\! d\boldsymbol{s}d\boldsymbol{y} \!\!\! \end{align}
 Using Lemma \ref{equality_lemma} and the fact that the joint distribution
$f_{S_{G},Y_{G}}$ is Gaussian: \begin{align}
= & \!\!\!\int \!\!\! \int \!\!\! f_{S,Y}(\boldsymbol{s},\boldsymbol{y})\left[\log f_{S|Y}(\boldsymbol{s},\boldsymbol{y})\!-\!\log f_{S_{G}|Y_{G}}(\boldsymbol{s},\boldsymbol{y})\right]d\boldsymbol{s}d\boldsymbol{y}\\
= & \int f_{Y}(\boldsymbol{y})\int f_{S|Y}(\boldsymbol{s},\boldsymbol{y}) \log\frac{ f_{S|Y}(\boldsymbol{s},\boldsymbol{y})}{f_{S_{G}|Y_{G}}(\boldsymbol{s},\boldsymbol{y})}d\boldsymbol{s}d\boldsymbol{y}\\
= &\mathcal{D}(f_{S|Y},f_{S_{G}|Y_{G}})\end{align}
 $\mathcal{D}\geq0$ with equality if and only if $\boldsymbol S\sim\mathcal{N}(\boldsymbol 0,{R}_{S}).$

\end{proof} 

Next, we present our main result on this topic:

 \begin{theorem}
\label{main_result} For a Gaussian vector source $\boldsymbol{X} \in \mathbb R^K$
and MSE distortion $d(\boldsymbol{x},\boldsymbol{ \hat x})=(\boldsymbol{x}-\boldsymbol{\hat x})^{T}(\boldsymbol{x}-\boldsymbol{\hat x})$,
the following holds:
\begin{equation}
R_{I}(D)=R_{U}(D)\label{eq}\end{equation}
 \end{theorem} 
 \begin{proof} Generally, $R_{I}(D)\geq R_{U}(D)$,
since independent reconstruction error is also uncorrelated. Note
that the uncorrelated error constraint dictates ${R}_{SX}={0}$
and the distortion constraint is $Tr({R}_{S})=D$. Lemma
\ref{main_lemma} states that under these constraints, for a Gaussian
source, Gaussian reconstruction error minimizes the mutual information
between the source and the reconstruction, i.e., $I(\boldsymbol{X}_{G},\boldsymbol{X}_{G}+\boldsymbol{S})$
achieves its minimum when $\boldsymbol{S}\sim\mathcal{N}(\boldsymbol{0},{R}_{S})$.
Then, $\boldsymbol{X}_{G}$ and $\boldsymbol{S}$ are uncorrelated
and jointly Gaussian and are, thereby, also independent. \end{proof}

We next pose the question: Are there cases where the best possible
vector quantizer at asymptotically high dimension that renders the
reconstruction error uncorrelated with the source, is necessarily a randomized
one? The next corollary answers in the affirmative, as is proved by the Gaussian case. 

\begin{corollary} For a Gaussian source, at asymptotically high
quantizer dimension, the quantizer that achieves minimum distortion
subject to the uncorrelated error constraint is necessarily a randomized
one. \end{corollary}

\begin{proof} From Theorem \ref{main_result}, the reconstruction
error for the Gaussian source subject to the uncorrelated error constraint
is independent of the source. No deterministic quantizer can render
the quantization noise independent from the source by definition;
hence, the optimal quantizer is a randomized one. \end{proof}

Note that our results hold only asymptotically, it is still an open
question whether or not they hold at finite dimensions. The numerical
results in the next section support the thesis that randomized quantizers are better  at
finite dimensions.

\section{Experimental Results}

In this section, we numerically compare the proposed quantizers to the conventional (uniform) dithered quantizer  and to the optimal quantizer, for a standard unit variance scalar Gaussian source. We implemented the proposed quantizers by numerically calculating the derived integrals. For that purpose, we sampled the distribution on a uniform grid. We
also imposed bounded support ($-3$ to $3$) i.e., we
numerically neglected the tails of the Gaussian. In this paper, we proposed three quantizers:

{\bf {Quantizer 1}}: Unconstrained randomized quantizer. This quantizer does not render the reconstruction error uncorrelated with the source.

{\bf {Quantizer 2}}: Constrained randomized quantizer which renders the quantization error uncorrelated with the source.

{\bf {Quantizer 3}}: Constrained deterministic quantizer which renders the quantization error uncorrelated with the source.

\begin{figure}
\centering \includegraphics[scale=0.65]{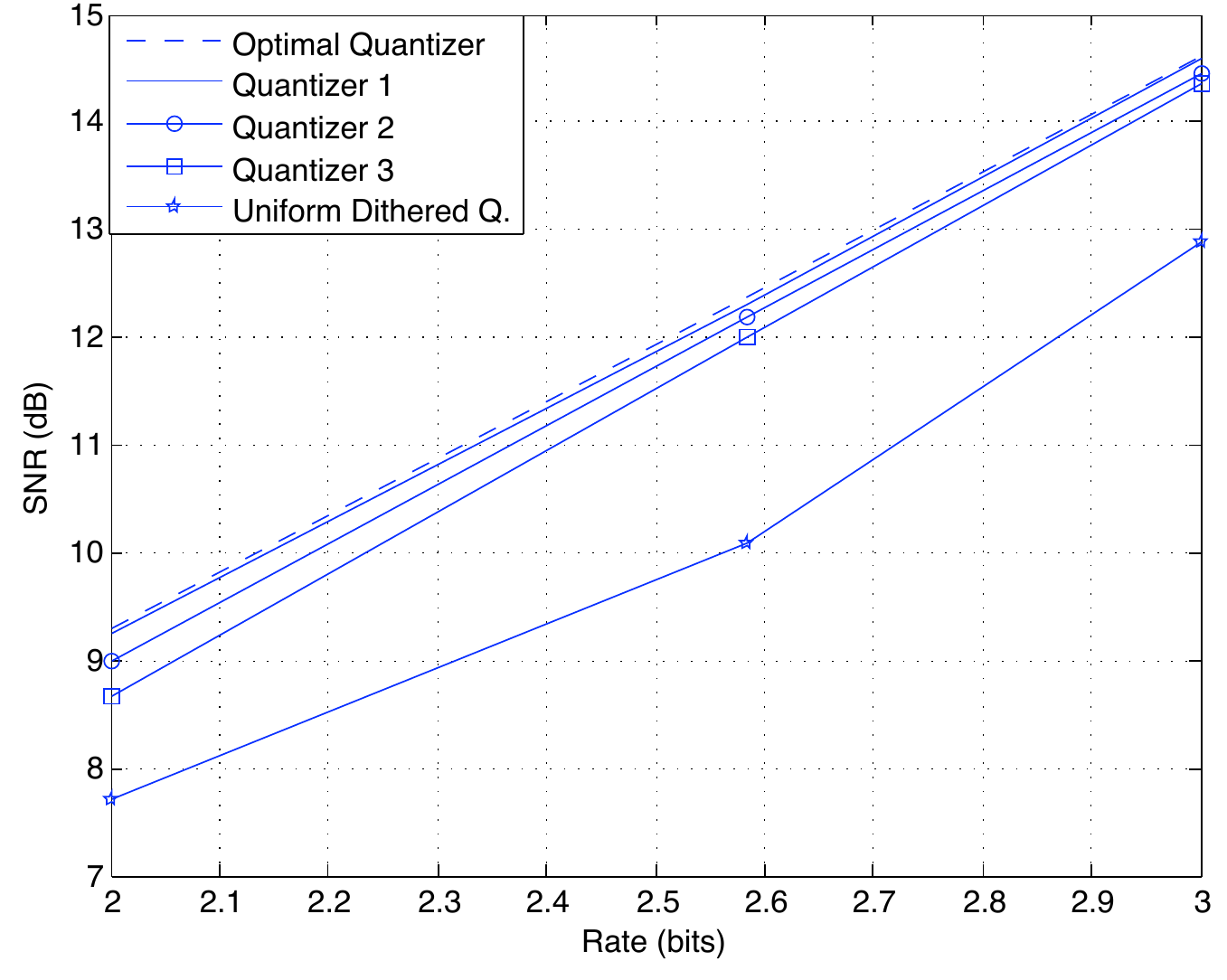} \caption{Performance comparison in terms of SNR versus rate for fixed rate quantization.}
\label{fig2} %
\end{figure}

\begin{figure}
\centering \includegraphics[scale=0.65]{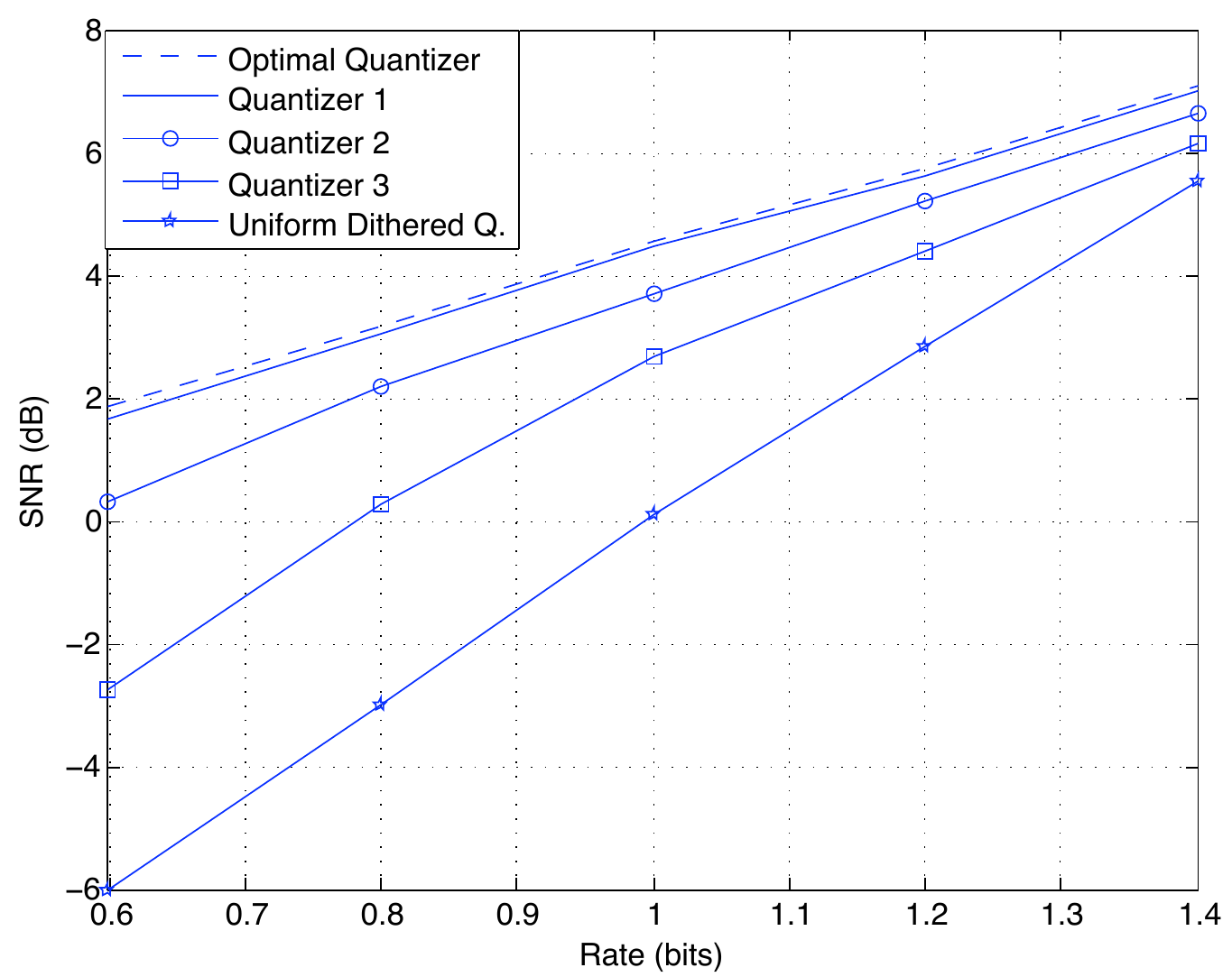} \caption{Performance comparison in terms of SNR versus rate for variable rate quantization.}
\label{fig3} %
\end{figure}

Figures 3 and 4 demonstrate the performance comparisons among quantizers for fixed and variable rates respectively. Note that for both fixed and variable rate, the optimal randomized quantizer performs very close to the optimal quantizer. However, it does not provide the statistical benefits of the other  quantizers. 

Note that for fixed rate, conventional (uniform) dithered quantization suffers significantly from the suboptimality of having equal quantization intervals irrespective of the rate region. 
However, at variable rate, the difference between the proposed and conventional dithered quantizer diminish at high rates, while at low rates the difference is quite significant. This is theoretically expected since at high rates, the optimal variable rate quantizer is very close to uniform, hence there is not much to gain from using a non-linear compressor-expander.  

For both fixed and variable rate, the constrained randomized quantizer outperforms  its deterministic counterpart, while both of them perform significantly better than the conventional dithered quantizer. Both of the proposed quantizers render quantization error {\it uncorrelated} with the source with low performance degradation while the dithered uniform quantizer renders  error {\it independent} of the source but (depending on the rate) at significant distortion penalty.

An additional benefit of the proposed random quantizers pertains to the correlation of the reconstruction errors  when correlated sources are quantized. The conventional dithered quantizer renders quantization error independent of the source hence, when two correlated sources are quantized with a dithered quantizer, the reconstruction errors are uncorrelated. For deterministic quantizers (Quantizer 3 and the optimal quantizer), the reconstruction error is a deterministic function of the source hence, intuitively randomized quantizers are expected to have lower reconstruction error correlation. Figures 5 and 6 demonstrate the correlation of the reconstruction error for different values of  source correlation for a bivariate Gaussian source for fixed and variable rate quantization respectively. These numerical results illustrate this intuitive conclusion:  randomization is significantly useful in decreasing the correlations of reconstruction errors. Specifically, the constrained randomized quantizer (Quantizer 2) yields extremely low error correlation, very close to that of the conventional dithered quantization. This property is particularly useful in practical applications such as image-video compression where white reconstruction error is preferred due to audio-visual considerations, see eg.  deblocking filters commonly used in video coding \cite{list2003adaptive}. Also note that, the unconstrained randomized quantizer (Quantizer 1) significantly decreases the error correlation compared to the optimal quantizer, with negligible sacrifice in rate distortion performance. Hence, this statistical benefit of randomization comes with no significant penalty. 

\begin{figure}
\centering \includegraphics[scale=0.65]{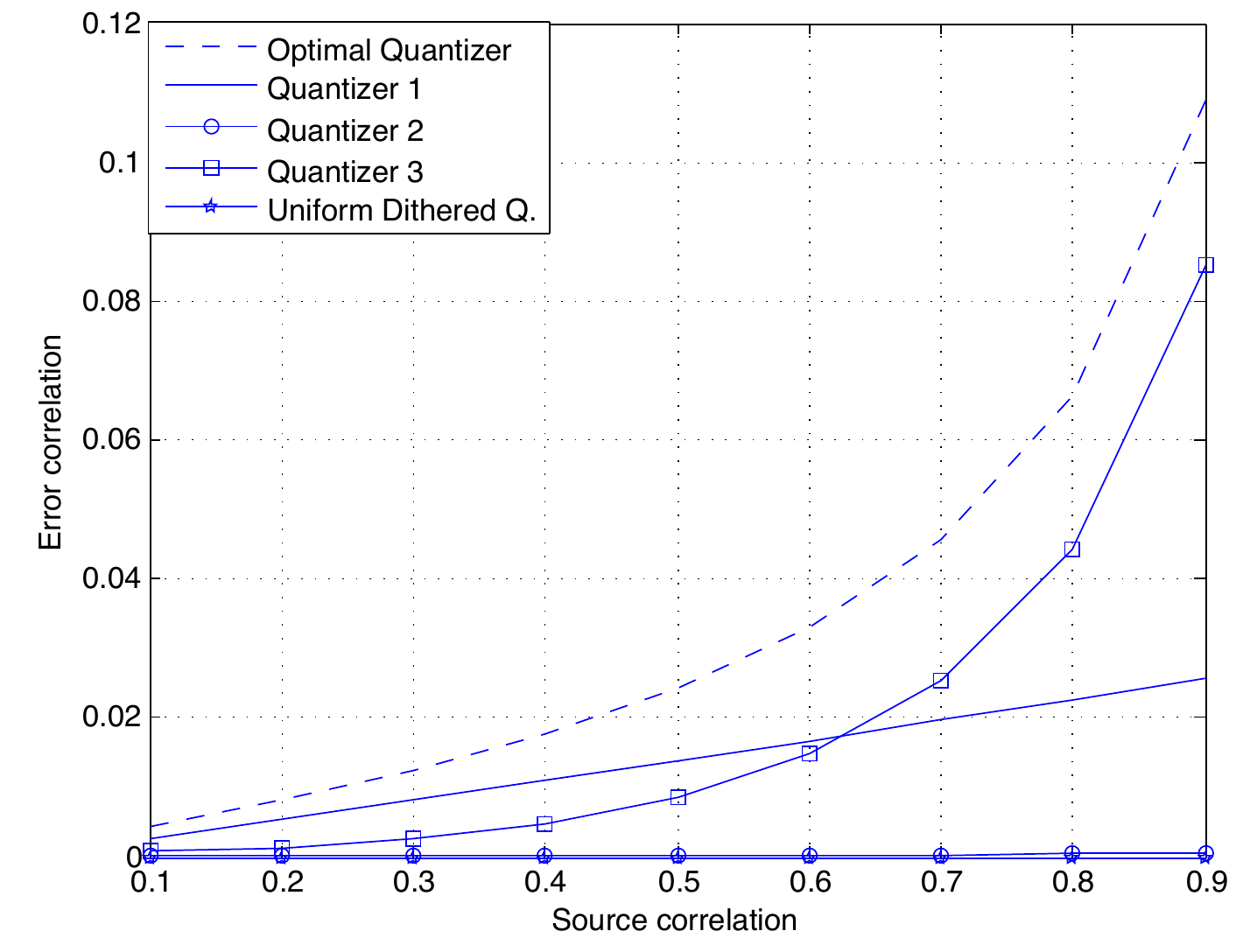} \caption{Correlation of the reconstruction error versus source correlation for fixed rate quantization at rate $R=2$bits/sample.}
\label{fig2} %
\end{figure}

\begin{figure}
\centering \includegraphics[scale=0.65]{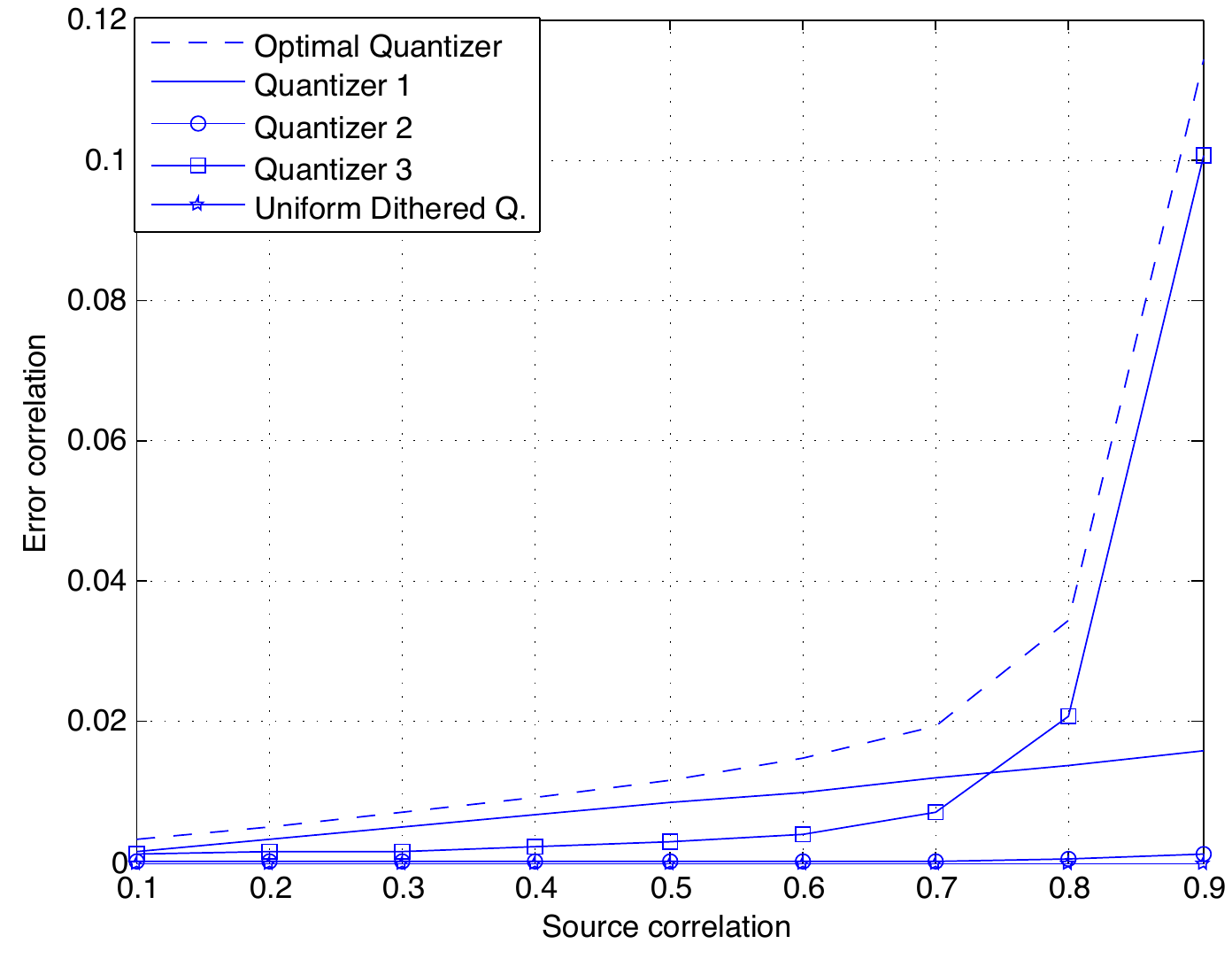} \caption{Correlation of the reconstruction error versus source correlation for variable rate quantization at rate $R=1.4$ bits/sample.}
\label{fig3} %
\end{figure}

Numerical comparisons show that the proposed quantization schemes
can significantly impact the design of compression systems such as
\cite{guleryuz,mihcak} where quantization error is assumed to be 
uncorrelated with the source. Note that the constrained randomized
quantization satisfies this assumption exactly and significantly outperforms
the conventional dithered quantization, which has been presented in
such prior work as the viable option to satisfy these assumptions.
In fact, as an alternative to the conventional dithered quantization that satisfies these assumptions at the considerable performance cost, we
derived additional quantization schemes that satisfy those
assumptions: constrained deterministic quantization and constrained nonuniform
random quantization. We also derived an unconstrained randomized quantizer, which performs almost as well as the optimal (deterministic) quantizer, yet offers perceptual benefits typical to dithered quantization. 

While it is difficult to prove, in general, the strict superiority
of these new quantizers over the conventional dithered quantizer,
we numerically show it, for both fixed and variable rate quantizers,
in Figures 3 and 4. Moreover, the numerical results motivate a theoretical
proposition: The optimal vector quantizer that renders the reconstruction
error orthogonal to the source is necessarily randomized. While we proved this
result at asymptotically high dimensions, it remains a conjecture  at finite dimensions, based on the numerical results in
this section.

\section{Discussion}

In this paper, we proposed a nonuniform randomized quantizer where dithering is performed in the companded domain to circumvent the problem of matching the dither
range to varying quantization intervals. The optimal compressor and expander mappings that minimize the mean square error are found via a novel numerical method.
Also, we discovered the connections between the optimal quantizer and the one whose
reconstruction error is constrained to be orthogonal to the source, for both deterministic and randomized quantization. The proposed constrained randomized quantization outperforms conventional dithered quantization and also the constrained deterministic quantizer proposed in this paper, while still satisfying the requirement that the reconstruction error be
uncorrelated with the source. Moreover, the proposed randomized quantizers significantly reduce the correlations across reconstruction errors when correlated sources, i.e., sources with memory, are quantized. 
We also showed that at asymptotically high dimensions, the MSE optimal vector quantizer designed for a vector Gaussian source,
which renders the reconstruction error uncorrelated with the source, must be a randomized quantizer. As  future work, we will investigate the applicability of this result to a broader class of sources where random encoding is not merely a tool to derive  rate-distortion bounds, but a necessary element in practical systems approaching such bounds.

 \bibliographystyle{IEEEbib}
\bibliography{ref}

\end{document}